\documentclass[11pt]{article}

\usepackage[margin=1in]{geometry}
\usepackage[T1]{fontenc}
\usepackage{times}
\usepackage{microtype}
\usepackage{natbib}
\usepackage{amsmath}
\usepackage{amssymb}
\usepackage{amsthm}
\usepackage{mathtools}
\usepackage{booktabs}
\usepackage{graphicx}
\usepackage{xcolor}
\usepackage{enumitem}
\usepackage{hyperref}
\usepackage[nameinlink,capitalise]{cleveref}

\hypersetup{
  colorlinks=true,
  linkcolor=blue!55!black,
  citecolor=green!40!black,
  urlcolor=blue!65!black,
  pdftitle={A Rank-One Popularity Component in Dot-Product Recommender Scores: Population Theory and Prior-Separation Evidence},
  pdfauthor={Yang Cheng},
  pdfsubject={Long-tail recommendation, representation anisotropy, and prior-adjusted softmax},
  pdfkeywords={recommender systems, anisotropy, long tail, softmax, representation degeneration}
}

\newtheorem{proposition}{Proposition}
\newtheorem{corollary}{Corollary}
\newtheorem{remark}{Remark}
\newcommand{\R}{\mathbb{R}}
\newcommand{\V}{\mathcal{V}}
\newcommand{\C}{\mathcal{C}}
\newcommand{\E}{\mathbb{E}}
\newcommand{\one}{\mathbf{1}}
\newcommand{\softmax}{\operatorname{softmax}}
\newcommand{\PMI}{\operatorname{PMI}}
\newcommand{\erank}{\operatorname{erank}}

\title{A Rank-One Popularity Component in Dot-Product Recommender Scores:\\
Population Theory and Prior-Separation Evidence}

\author{
Yang Cheng\\
Independent Researcher\\
\texttt{wxwidget@gmail.com}
}
\date{}

\begin{document}

\maketitle

\begin{abstract}
Representation anisotropy in recommenders is often attributed to Transformer architecture. We show that a more general source lies in the conditional training distribution. For any encoder with a dot-product softmax decoder, the population-optimal score decomposes into pointwise mutual information, an item-marginal term $\log p(i)$, and a context-dependent gauge. After centering, the item marginal forms a context-shared rank-one score component, while time-varying marginals induce a low-rank popularity subspace. Transfer from scores to embeddings depends on factorization geometry, so this result does not imply universal embedding collapse.

We test the mechanism with synthetic data and public Alibaba/Tianchi interaction logs. In a matched intervention, separating $\log p(i)$ from the learned dot product reduces the measured context-shared popularity-aligned score energy by $98.6\%$; permutation tests confirm that the reduction is specific to the empirical popularity direction. The intervention also exposes an accuracy--coverage trade-off and incurs a small held-out likelihood cost. These results recast a class of apparent representation degeneration as a decoder-level consequence of long-tailed item marginals, independent of encoder architecture, while distinguishing the population mechanism from its non-identifying empirical proxy.
\end{abstract}

\section{Introduction}

Sequential recommendation systems encode a user's interaction history into a vector and retrieve candidate items by dot-product similarity. Transformer encoders such as SASRec and BERT4Rec have become standard choices for constructing the context representation \citep{kang2018sasrec,sun2019bert4rec}, but the same decoder is also used with recurrent, convolutional, pooling, and graph-based encoders. Across these systems, learned item representations can become anisotropic: singular values decay rapidly, unrelated items have unexpectedly high cosine similarity, and a small number of directions explain most embedding energy \citep{gao2019representation,qiu2022duorec}.

Prior work has described this effect as representation degeneration and has successfully mitigated it with contrastive regularization \citep{qiu2022duorec}. The empirical observation is important, but the usual architectural phrasing can obscure the underlying mechanism. The equation that produces the problematic frequency term does not contain a Transformer block. It contains a dot-product softmax trained on a long-tailed item distribution.

This paper isolates one such mechanism. Let $h_c$ be the representation produced by an arbitrary encoder for context $c$, and let $e_i$ be the embedding of item $i$. A softmax decoder uses $h_c^\top e_i$ as its logit. At the population optimum,
\begin{equation}
  h_c^\top e_i
  = \log p(i\mid c) + a_c
  = \PMI(c,i) + \log p(i) + a_c,
  \label{eq:intro-decomposition}
\end{equation}
where $a_c$ is an item-independent shift. The semantic term $\PMI(c,i)$ varies with context. The popularity term $\log p(i)$ does not. If the decoder has no separate prior channel, the factorized dot product must represent both.

The principle of separating popularity from compatibility appeared in our earlier dynamic popularity-aware (DPA) formulation for contrastive recommendation \citep{lin2021dpa}. Under noise-contrastive estimation (NCE), choosing the time-local item marginal as the noise distribution makes the optimal dot product estimate time-local PMI up to a constant; adding the contemporaneous log marginal at inference recovers the interaction score. That work targeted dynamic matching accuracy and validated the strategy on Alibaba interaction data. It did not ask what geometric structure a full-softmax model learns when the prior remains inside the dot product, nor when score structure transfers to embedding anisotropy. We address both questions and return to public Alibaba/Tianchi competition logs with score-level diagnostics designed for the proposed mechanism.

Stacking logits for all contexts and items makes the geometry visible. The popularity contribution is $\one f^\top$, where $f_i=\log p(i)$, and therefore has rank one. A Zipfian marginal makes the centered vector $f$ high-variance. When this component dominates the context-dependent semantic matrix, the score matrix develops a large first singular value. Under a balanced minimum-norm factorization, the same concentration appears in both context and item factors. The encoder architecture affects the semantic residual and optimization path, but it is not required for the rank-one pressure to exist.

We also clarify the role of non-target, often called negative-sample, gradients. For one training pair $(c,y)$,
\begin{equation}
  \nabla_{e_k}\ell(c,y)
  = \bigl(q_\theta(k\mid c)-\mathbb{I}[k=y]\bigr)h_c.
\end{equation}
Rare items receive few positive semantic corrections and many small non-target updates. In finite data this creates noisy, shared-direction trajectories. At the population optimum, however, positive and non-target terms cancel in expectation. The gradient account therefore describes how the statistical bias is realized during training, rather than a second root cause.

Our contributions are:
\begin{enumerate}[leftmargin=1.5em]
  \item For full-softmax dot-product recommenders, we identify the item marginal's gauge-invariant rank-one score component and generalize $T$ time-varying marginals to a rank-at-most-$T$ temporal popularity subspace.
  \item We establish a spectral-gap bound and show when score concentration transfers to item embeddings under balanced minimum-norm factorization.
  \item We reconcile the population-optimum and negative-gradient explanations, making explicit what each does and does not prove.
  \item Building on the prior-separation principle used in DPA, we derive a continuous prior-adjusted full-softmax objective and distinguish target-shaping NCE from proposal-corrected sampled softmax.
  \item We provide two reproducible matched studies: a controlled synthetic stress test with gauge-correct factorization diagnostics, and a three-seed study on Alibaba/Tianchi logs that measures an operational popularity-aligned score component, tests its direction specificity, and audits likelihood, retrieval, and exposure trade-offs.
\end{enumerate}

The resulting conclusion is deliberately narrower than ``Transformers are anisotropic'' and more actionable: a long-tail item marginal contributes a low-rank component to population-optimal dot-product softmax scores regardless of which encoder produced $h_c$. This is an existence claim about the score target. Architecture, optimization, and the data-generating exposure policy determine how strongly a corresponding direction appears in learned scores and embeddings.

\section{Related Work}

\paragraph{Representation anisotropy.}
Anisotropy has been documented in contextual language representations \citep{ethayarajh2019contextual} and in likelihood-trained language-generation embeddings, where vectors can occupy a narrow cone \citep{gao2019representation}. DuoRec identified analogous degeneration in sequential recommendation, connected rare-item updates to shared context directions, and used contrastive regularization to improve uniformity \citep{qiu2022duorec}. Our analysis starts from the same softmax decomposition but separates three notions that are sometimes conflated: spectral concentration of the optimal score matrix, anisotropy of a particular factorization, and finite-sample gradient dynamics.

\paragraph{Sequential recommendation.}
GRU4Rec introduced recurrent next-item modeling for session recommendation \citep{hidasi2016gru4rec}. SASRec and BERT4Rec later adopted self-attention and bidirectional masked modeling \citep{kang2018sasrec,sun2019bert4rec}. These encoders differ substantially, yet their retrieval heads commonly use item embeddings and dot products. This shared decoder motivates an encoder-independent analysis of the popularity channel.

\paragraph{Dynamic popularity separation.}
Our earlier DPA method used time-local item popularity as the NCE noise distribution, thereby learning a time-local PMI score, and restored current popularity during retrieval \citep{lin2021dpa}. Its contribution was an operational training--testing strategy for dynamic matching, supported by experiments on Alibaba e-commerce interactions. The present paper studies a different but connected question: under a full-softmax dot-product decoder, what spectral structure is created when popularity is internalized by the factors? We supply the gauge-invariant rank and spectrum analysis absent from DPA, then use the same data source for a matched score-geometry intervention rather than treating DPA's ranking gains as evidence for anisotropy.

\paragraph{Observed popularity and exposure bias.}
Recommendation logs are selected by user self-selection and prior serving policies, so an observed item marginal need not equal latent relevance \citep{schnabel2016treatments}. Work on missing-not-at-random implicit feedback uses propensity-aware objectives to target relevance under nonuniform exposure \citep{saito2020unbiased}, while causal embedding methods separate user interest from conformity toward popular items \citep{zheng2021dice}. Our estimand is different: we analyze the score geometry induced by the observed training distribution rather than identifying an exposure-free preference distribution. This boundary is why the empirical log-frequency direction is operational and policy-dependent.

\paragraph{Alignment and uniformity.}
\citet{wang2020uniformity} characterize contrastive representation learning through positive-pair alignment and hyperspherical uniformity. Recommendation methods including DuoRec, SimGCL, and DirectAU use related objectives to improve representation geometry \citep{qiu2022duorec,yu2022simgcl,wang2022directau}. We view uniformity as a useful geometric regularizer, but not as a substitute for modeling the item prior that enters the population-optimal logits.

\paragraph{Long-tail correction.}
Logit adjustment provides a statistically grounded treatment of long-tailed label marginals by adding or subtracting log priors during training or inference \citep{menon2021logit}. Balanced Softmax similarly modifies the normalization to account for class frequency \citep{ren2020balanced}. These ideas are usually presented for classification. We characterize their geometric consequence in factorized recommendation: prior adjustment removes a rank-one popularity component that would otherwise consume embedding capacity.

\paragraph{Matrix factorization and sampling.}
The connection between predictive objectives and matrix factorization has a long history, including the PMI interpretation of word embeddings \citep{levy2014neural}. Large-vocabulary models often rely on importance or sampled-softmax approximations \citep{bengio2008adaptive,jean2015using}. The role of the sampling distribution must be stated explicitly: proposal correction is required when approximating a fixed full-softmax target, whereas NCE methods such as DPA deliberately use the noise distribution to define a PMI-like target and then restore the prior at inference.

\section{Problem Setup}
\label{sec:setup}

Let $\C=\{1,\ldots,m\}$ denote a finite set of contexts and $\V=\{1,\ldots,n\}$ an item vocabulary. The data distribution is $\pi(c,i)$, with conditionals $\pi(i\mid c)$ and item marginal $\pi_i=\pi(i)$. We use \emph{popularity} operationally for this observed event marginal. It can combine exposure policy, availability, interface position, repeated behavior, and preference; the analysis does not identify an intrinsic item attribute or a causal preference distribution. We assume $\pi(i\mid c)>0$ for notational simplicity; zero-probability events can be handled by restricting to the support or by smoothing.

An arbitrary encoder $g_\theta$ maps context $c$ to $h_c=g_\theta(c)\in\R^d$. Each item has an embedding $e_i\in\R^d$. The decoder is
\begin{align}
  z_{ci} &= h_c^\top e_i, \\
  q_\theta(i\mid c)
  &= \frac{\exp(z_{ci})}{\sum_{j=1}^{n}\exp(z_{cj})}.
\end{align}
The population cross-entropy is
\begin{equation}
  \mathcal{L}_{\mathrm{CE}}
  = -\E_{(c,i)\sim\pi}\log q_\theta(i\mid c).
\end{equation}

We distinguish three objects:
\begin{enumerate}[leftmargin=1.5em]
  \item the unconstrained optimal logit matrix $Z^*\in\R^{m\times n}$;
  \item a finite-rank factorization $Z=HE^\top$, where rows of $H$ and $E$ are context and item representations;
  \item the optimization trajectory used to learn one such factorization.
\end{enumerate}
The distinction matters because softmax logits are invariant to row shifts and matrix factorizations are not unique. A claim about $Z^*$ does not automatically identify the geometry of every possible $E$.

\subsection{Anisotropy measures}

For an item embedding matrix $E\in\R^{n\times d}$ with singular values $s_1\geq\cdots\geq s_d\geq0$, we use
\begin{equation}
  R_1(E)=\frac{s_1^2}{\sum_{k=1}^{d}s_k^2}
\end{equation}
as the leading-direction energy ratio. We also report effective rank,
\begin{equation}
  \erank(E)
  = \exp\left(-\sum_{k=1}^{d}\rho_k\log\rho_k\right),
  \qquad
  \rho_k=\frac{s_k^2}{\sum_j s_j^2}.
\end{equation}
Large $R_1$ and small effective rank indicate spectral anisotropy. Mean pairwise cosine similarity measures cone concentration, but unlike the spectrum it is sensitive to translation and softmax gauge. We therefore treat cosine concentration as a related empirical symptom rather than an equivalent mathematical definition.

To remove item-independent row shifts, define the item-centering projector
\begin{equation}
  P_n=I_n-\frac{1}{n}\one_n\one_n^\top.
\end{equation}
The row-centered score matrix is $\widetilde Z=ZP_n$. This quantity is invariant to replacing any row $z_c$ with $z_c+a_c\one_n$.

\section{The Item Marginal Contributes a Rank-One Score Component}
\label{sec:theory}

\subsection{Population-optimal logits}

\begin{proposition}[Softmax optimum]
\label{prop:softmax-optimum}
For every context $c$, any unconstrained minimizer of the population cross-entropy satisfies
\begin{equation}
  z^*_{ci}=\log \pi(i\mid c)+a_c,
  \label{eq:optimal-logit}
\end{equation}
where $a_c\in\R$ is arbitrary and independent of $i$.
\end{proposition}

\begin{proof}
For fixed $c$, cross-entropy decomposes as
\begin{equation}
  -\sum_i\pi(i\mid c)\log q(i\mid c)
  = H(\pi(\cdot\mid c))
  +D_{\mathrm{KL}}\!\left(\pi(\cdot\mid c)\,\|\,q(\cdot\mid c)\right).
\end{equation}
It is minimized exactly when $q(\cdot\mid c)=\pi(\cdot\mid c)$. Inverting softmax gives \cref{eq:optimal-logit}, with the additive constant expressing softmax's row-shift invariance.
\end{proof}

Using $\pi(i\mid c)=\pi(c,i)/\pi(c)$,
\begin{equation}
  \log\pi(i\mid c)
  = \PMI(c,i)+\log\pi_i.
\end{equation}
Therefore
\begin{equation}
  z^*_{ci}
  = \PMI(c,i)+\log\pi_i+a_c.
  \label{eq:pmi-popularity}
\end{equation}
This identity is independent of the encoder that produced $h_c$.

\subsection{Gauge-invariant rank-one structure}

Let $M_{ci}=\PMI(c,i)$, $f_i=\log\pi_i$, and $a=(a_1,\ldots,a_m)^\top$. In matrix form,
\begin{equation}
  Z^*=M+\one_m f^\top+a\one_n^\top.
\end{equation}
Multiplying by $P_n$ removes the gauge term:
\begin{equation}
  \widetilde Z^*
  = \widetilde M+\one_m\widetilde f^\top,
  \qquad
  \widetilde M=MP_n,
  \quad
  \widetilde f=P_nf.
  \label{eq:centered-score}
\end{equation}
The popularity variation $F=\one_m\widetilde f^\top$ has rank one and sole nonzero singular value
\begin{equation}
  \sigma_1(F)=\sqrt{m}\,\|\widetilde f\|_2.
\end{equation}

\begin{proposition}[Popularity-induced spectral gap]
\label{prop:spectral-gap}
For the row-centered optimum in \cref{eq:centered-score},
\begin{align}
  \sigma_1(\widetilde Z^*)
  &\geq \sqrt{m}\,\|\widetilde f\|_2-\|\widetilde M\|_2,
  \label{eq:lower-top}\\
  \sigma_2(\widetilde Z^*)
  &\leq \|\widetilde M\|_2.
  \label{eq:upper-second}
\end{align}
Consequently, if $\sqrt{m}\|\widetilde f\|_2>2\|\widetilde M\|_2$, the leading singular value is strictly separated from the second.
\end{proposition}

\begin{proof}
The popularity matrix $F=\one_m\widetilde f^\top$ is rank one. Weyl's singular-value inequalities give
\begin{equation}
  \sigma_1(F+\widetilde M)
  \geq \sigma_1(F)-\sigma_1(\widetilde M)
\end{equation}
and
\begin{equation}
  \sigma_2(F+\widetilde M)
  \leq \sigma_2(F)+\sigma_1(\widetilde M)
  =\|\widetilde M\|_2.
\end{equation}
Substitution yields the result.
\end{proof}

The proposition states a sufficient condition, not a claim that every long-tailed dataset must be degenerate. Strong context-specific PMI structure can compete with the prior. The key point is that the prior contributes a coherent direction repeated across all contexts, while semantic variation is distributed over many directions.

It is useful to summarize the competition by the dimensionless ratio
\begin{equation}
  \rho_{\mathrm{pop}}
  =\frac{\sqrt{m}\,\|\widetilde f\|_2}{\|\widetilde M\|_2}.
  \label{eq:popularity-ratio}
\end{equation}
Then $\rho_{\mathrm{pop}}>2$ is sufficient for the strict separation in \cref{prop:spectral-gap}. If the data-generating logit contains $\alpha\log\pi_i$, the numerator is multiplied by $|\alpha|$.

\begin{corollary}[Zipfian scaling]
\label{cor:zipf}
Suppose item marginals follow $\pi_i=i^{-\zeta}/H_{n,\zeta}$ for rank $i\in\{1,\ldots,n\}$ and exponent $\zeta>0$. Then
\begin{equation}
  \|\widetilde f\|_2
  = \zeta\left[\sum_{i=1}^{n}
  \left(\log i-\frac{1}{n}\sum_{j=1}^{n}\log j\right)^2\right]^{1/2}
  =\Theta(\zeta\sqrt{n}).
\end{equation}
Thus the rank-one singular value scales as $\Theta(\zeta\sqrt{mn})$.
\end{corollary}

The derivation follows immediately because the normalizing constant in $\log\pi_i=-\zeta\log i-\log H_{n,\zeta}$ disappears under centering. The asymptotic variance of $\log i$ under uniformly distributed ranks is constant.

\subsection{Time-varying marginals form a temporal subspace}

The static rank-one result extends naturally when item popularity changes over time, as in dynamic recommendation \citep{lin2021dpa}. Let each context $c$ belong to one of $T$ periods, denoted $t(c)$, and let $\pi_{t,i}$ be the item marginal in period $t$. The semantic matrix now uses time-local mutual information, $M_{ci}=\PMI_{t(c)}(c,i)$. Define
\begin{equation}
  f_t=(\log\pi_{t,1},\ldots,\log\pi_{t,n})^\top,
  \qquad
  \widetilde f_t=P_nf_t,
\end{equation}
and collect the centered log marginals in
\begin{equation}
  L=[\widetilde f_1,\ldots,\widetilde f_T]\in\R^{n\times T}.
\end{equation}
Let $G\in\{0,1\}^{m\times T}$ be the period-membership matrix with $G_{ct}=\mathbb{I}[t(c)=t]$.

\begin{proposition}[Dynamic popularity subspace]
\label{prop:dynamic-popularity}
Suppose the population conditional for context $c$ uses the contemporaneous marginal $\pi_{t(c),i}$. After item centering, the optimal score matrix satisfies
\begin{equation}
  \widetilde Z^*
  =\widetilde M+GL^\top,
  \label{eq:dynamic-centered-score}
\end{equation}
where $\widetilde M$ contains centered time-local PMI terms. The dynamic popularity component obeys
\begin{equation}
  \operatorname{rank}(GL^\top)
  \leq \operatorname{rank}(L)
  \leq \min(T,n-1).
  \label{eq:dynamic-popularity-rank}
\end{equation}
If all periods share the same marginal, this component reduces to the static rank-one matrix in \cref{eq:centered-score}.
\end{proposition}

\begin{proof}
For every context $c$, the corresponding row of $GL^\top$ is $\widetilde f_{t(c)}^\top$, so the time-local PMI decomposition gives \cref{eq:dynamic-centered-score}. The rank inequality follows from the rank of a matrix product. Since every column of $L$ lies in the $(n-1)$-dimensional centered subspace, $\operatorname{rank}(L)\leq\min(T,n-1)$. When all columns of $L$ are equal, $GL^\top=\one_m\widetilde f^\top$.
\end{proof}

Thus temporal popularity is rank one within each period and occupies the span of the centered log-marginal trajectories globally. The exact bound can be as large as $T$; a smaller effective rank requires additional empirical structure, such as correlated or smoothly evolving marginals. DPA can be interpreted as estimating this time-local prior through its NCE noise distribution and restoring it at retrieval, while the present result identifies the score subspace that such separation prevents the semantic factors from absorbing.

\subsection{From score spectrum to embedding spectrum}

The factorization $HE^\top=Z$ is non-identifiable: for any invertible $A$, $(HA)(EA^{-\top})^\top=HE^\top$. We therefore state the transfer result under the balanced minimum-norm factorization associated with symmetric $\ell_2$ regularization.

\begin{proposition}[Balanced factorization transfer]
\label{prop:balanced-factorization}
Let $\widetilde Z=U\Sigma V^\top$ have rank $r$. Among exact rank-$r$ factorizations, consider
\begin{equation}
  \min_{H,E:\,HE^\top=\widetilde Z}
  \frac{1}{2}\left(\|H\|_F^2+\|E\|_F^2\right).
  \label{eq:balanced-objective}
\end{equation}
Every minimizer is, up to a common orthogonal transform,
\begin{equation}
  H=U\Sigma^{1/2}R,
  \qquad
  E=V\Sigma^{1/2}R.
\end{equation}
Hence the singular values of $E$ are $\sqrt{\sigma_k(\widetilde Z)}$, and
\begin{equation}
  R_1(E)
  =\frac{\sigma_1(\widetilde Z)}{\sum_{k=1}^{r}\sigma_k(\widetilde Z)}.
  \label{eq:embedding-energy}
\end{equation}
\end{proposition}

A proof is included in \cref{app:balanced-proof}. This proposition explains why a strong rank-one target can appear as an item-embedding principal component under symmetric weight decay or implicit balancing. It does not imply that every algebraically valid factorization has the same cosine geometry.

\paragraph{Scope of the transfer result.}
The proposition characterizes exact minimum-norm factors, not every trajectory of a finite-time optimizer. Adaptive preconditioning, asymmetric regularization, encoder constraints, and optimization error can produce unbalanced factors even when the score matrix is concentrated. Related implicit-regularization results establish balancing or low-complexity bias only under specific matrix-factorization dynamics \citep{gunasekar2017implicit,arora2019implicit}; they do not directly prove that Adam reaches \cref{eq:balanced-objective} in our softmax model. We therefore treat balancedness as a measurable transfer condition and evaluate it after removing item gauge and inactive factor dimensions in \cref{sec:experiments}.

\begin{remark}[What the theory attributes to the encoder]
The encoder determines which semantic matrices $M$ are realizable and how training reaches a factorization. Transformer, GRU, CNN, and pooling encoders may therefore differ quantitatively, and their dynamics may amplify or suppress the observed concentration. However, none is needed for the popularity matrix $\one_m\widetilde f^\top$ to appear. The existence of this statistical channel is decoder- and data-dependent; its empirical magnitude is not encoder-independent.
\end{remark}

\section{Negative Gradients Are an Optimization View}
\label{sec:optimization}

For one observed pair $(c,y)$, let
\begin{equation}
  \ell(c,y)=-\log q_\theta(y\mid c).
\end{equation}
Differentiating with respect to item embedding $e_k$ gives
\begin{equation}
  \nabla_{e_k}\ell(c,y)
  =\left(q_\theta(k\mid c)-\mathbb{I}[k=y]\right)h_c.
  \label{eq:item-gradient}
\end{equation}
The corresponding SGD update is
\begin{equation}
  \Delta e_k
  =\eta\left(\mathbb{I}[k=y]-q_\theta(k\mid c)\right)h_c.
  \label{eq:item-update}
\end{equation}
When $k=y$, the item moves toward $h_c$ by $\eta(1-q_k)h_c$. When $k\neq y$, it moves away by $\eta q_kh_c$. The latter is the non-target or negative gradient. In a full softmax, every non-target item contributes; with sampled softmax, only sampled items do.

Taking expectation over $y\sim\pi(\cdot\mid c)$ yields
\begin{equation}
  \E_{y\mid c}[\Delta e_k]
  =\eta\left(\pi(k\mid c)-q_\theta(k\mid c)\right)h_c.
  \label{eq:expected-update}
\end{equation}
At the population optimum $q_\theta=\pi$, this expectation is zero. Negative updates therefore do not by themselves establish an additional asymptotic force beyond the likelihood objective. They are one side of the same balance that produced \cref{eq:pmi-popularity}; finite-time, sampled, or architecture-specific dynamics can nevertheless remain important away from calibration.

Why, then, do rare items show shared-direction trajectories? First, an item with marginal probability $\pi_k$ appears positively only about $N\pi_k$ times in $N$ observations. The relative sampling error of that count is approximately
\begin{equation}
  \frac{\sqrt{N\pi_k(1-\pi_k)}}{N\pi_k}
  \approx \frac{1}{\sqrt{N\pi_k}},
\end{equation}
which is large for tail items. Second, before calibration, the accumulated non-target term is
\begin{equation}
  -\eta\sum_t q_\theta(k\mid c_t)h_{c_t}.
\end{equation}
If context representations have a nonzero mean $\mu_h=\E[h_c]$, many rare items receive a common component approximately proportional to $-\mu_h$, while too few positive events arrive to provide item-specific semantic corrections. Sampled negatives can amplify this effect when their proposal distribution overexposes popular contexts or items.

This account supports two conclusions. The negative-gradient explanation is useful for understanding finite-data training and motivates sampling correction. It should not be presented as an independent proof that a Transformer architecture must collapse, nor should population cancellation be read as evidence that finite-time negative-gradient effects are negligible.

\section{Derived Intervention and Complementary Remedies}
\label{sec:mitigation}

The score decomposition provides a full-softmax geometric justification for a principle already used in dynamic NCE recommendation: separate the item prior from semantic compatibility. Corrected sampling, tail supervision, and geometric regularization address related finite-sample or geometric effects, but they are complementary recommendations rather than consequences of the same theorem.

\subsection{Prior-adjusted semantic logits}

Let $s_{ci}=h_c^\top e_i$ denote a semantic score. Define
\begin{equation}
  q_\tau(i\mid c)
  =\frac{\exp\left(s_{ci}+\tau\log\pi_i\right)}
  {\sum_j\exp\left(s_{cj}+\tau\log\pi_j\right)},
  \label{eq:prior-adjusted}
\end{equation}
where $\tau\in[0,1]$ controls prior separation.

\begin{proposition}[Optimal prior-adjusted score]
\label{prop:prior-adjusted}
The unconstrained population optimum of \cref{eq:prior-adjusted} satisfies
\begin{equation}
  s^*_{ci}
  =\PMI(c,i)+(1-\tau)\log\pi_i+a_c.
  \label{eq:prior-adjusted-optimum}
\end{equation}
In particular, $\tau=1$ removes the item marginal from the semantic score.
\end{proposition}

\begin{proof}
By \cref{prop:softmax-optimum}, the complete adjusted logit must equal $\log\pi(i\mid c)+a_c$. Subtracting $\tau\log\pi_i$ and applying the PMI decomposition yields \cref{eq:prior-adjusted-optimum}.
\end{proof}

An equivalent implementation uses
\begin{equation}
  z_{ci}=h_c^\top e_i+b_i,
  \qquad b_i\approx\tau\log\pi_i.
\end{equation}
The bias should be fixed or regularized toward the prior. If $b_i$ and the embedding factors are both unconstrained, the split is not identifiable and popularity can leak back into $e_i$.

Training-time and post-hoc signs are easy to confuse. During prior-adjusted training, $+\tau\log\pi_i$ is supplied outside the semantic score, so $s_{ci}$ need not relearn it. For a conventionally trained model whose logits already contain the prior, a debiased retrieval score instead uses $z_{ci}-\tau\log\pi_i$.

\subsection{Connection to dynamic NCE separation}

Prior-adjusted full softmax and DPA's dynamic NCE procedure implement the same separation principle through different estimators. For an NCE objective with $K$ negatives drawn from $q_t^{\mathrm{noise}}(i)$, the unconstrained optimal score has the form
\begin{equation}
  s^*_{\mathrm{NCE}}(c,i)
  =\log\pi_t(i\mid c)-\log q_t^{\mathrm{noise}}(i)-\log K.
  \label{eq:nce-optimum}
\end{equation}
When $q_t^{\mathrm{noise}}(i)=\pi_{t,i}$, this becomes time-local $\PMI_t(c,i)-\log K$. At the exact optimum, the conditional log score is recovered by
\begin{equation}
  \log\pi_t(i\mid c)
  =s^*_{\mathrm{NCE}}(c,i)
  +\log q_t^{\mathrm{noise}}(i)+\log K.
  \label{eq:nce-recovery}
\end{equation}
DPA used this identity operationally with time-local noise and a recent popularity estimate at retrieval \citep{lin2021dpa}. The $\tau=1$ case of \cref{prop:prior-adjusted} is its full-softmax analogue at the population-score level: both reserve the dot product for a PMI-like residual and place popularity in an explicit channel. They are not identical training estimators, and dynamic prior misspecification affects them differently.

\subsection{Negative sampling: correction versus target shaping}

When full softmax is impractical, let negatives be sampled from proposal $q(i)$. A useful mixture is
\begin{equation}
  q(i)=\rho\frac{1}{n}
  +(1-\rho)\frac{\pi_i^\gamma}{\sum_j\pi_j^\gamma},
  \qquad 0<\gamma<1,
  \label{eq:mixture-sampling}
\end{equation}
which preserves hard popular negatives while increasing tail coverage. A sampled-softmax logit should include the standard proposal correction
\begin{equation}
  \widetilde z_{ci}=z_{ci}-\log(Kq(i))
\end{equation}
for $K$ samples, or use another estimator with an explicitly characterized target. Changing $q$ without correction changes the learned distribution rather than merely reducing variance.

Two uses of $q$ should therefore not be conflated. If sampled softmax is intended to approximate a fixed full-softmax likelihood, proposal correction is required. If NCE deliberately defines the target relative to $q_t^{\mathrm{noise}}$, as in \cref{eq:nce-optimum}, the changed target is intentional and must be paired with an explicit recovery rule such as \cref{eq:nce-recovery}. The error is not target shaping itself, but changing the target without stating or recovering it.

Tail-positive supervision can be strengthened with clipped weights
\begin{equation}
  w_i=\operatorname{clip}
  \left((\pi_i+\epsilon)^{-\alpha},w_{\min},w_{\max}\right),
\end{equation}
or with content-based and sequence-level augmentation. Clipping is important because unbounded inverse-frequency weights can make the rarest observations dominate optimization.

\subsection{Geometric regularization}

After prior separation, uniformity can reduce residual cone concentration. With normalized item vectors $\widehat e_i=e_i/\|e_i\|$, one option is \citep{wang2020uniformity}
\begin{equation}
  \mathcal{L}_{\mathrm{unif}}
  =\log\E_{i\neq j}
  \exp\left(-t\|\widehat e_i-\widehat e_j\|_2^2\right).
\end{equation}
Contrastive objectives can impose a similar pressure on sequence representations and indirectly reshape item embeddings \citep{qiu2022duorec}. Covariance or spectral penalties offer alternatives. Normalization alone only removes norm variation; it does not prevent all vectors from sharing the same angular direction.

This distinction motivates the practical engineering order
\begin{equation}
  \boxed{\text{prior separation}}
  \;\rightarrow\;
  \boxed{\text{target-aware sampling}}
  \;\rightarrow\;
  \boxed{\text{geometric regularization}}.
\end{equation}
The first term addresses the population target and is derived from \cref{prop:prior-adjusted}. Target-aware sampling means either proposal correction for a fixed likelihood or explicit NCE target shaping with recovery. The final term addresses residual geometry. The sequence is a design recommendation, not a theorem of optimality.

\section{Empirical Studies}
\label{sec:experiments}

We conduct two matched studies. The first isolates the population mechanism and conditional score-to-factor transfer with a synthetic generator. The second tests the predicted score channel on temporally split Alibaba/Tianchi logs and intervenes on whether the dot product must internalize the item marginal.

\subsection{Controlled synthetic study}

\paragraph{Data generation.}

We sample $m=64$ context factors and $n=256$ item factors from isotropic Gaussians. Their dot product produces a centered semantic score matrix $M$ of rank at most eight. Item marginals follow a Zipf law with exponent $1.15$. For popularity strength $\alpha$, the data distribution is
\begin{equation}
  \pi_\alpha(i\mid c)
  \propto
  \exp\left(M_{ci}+\alpha\log\pi_i\right).
  \label{eq:synthetic-data}
\end{equation}
We sweep $\alpha\in\{0,0.5,1,1.5,2\}$ and repeat each experiment for five random seeds.

We optimize exact expected cross-entropy rather than drawing observations. This removes mini-batch and negative-sampling noise, isolating the population-target effect. Both models use ten-dimensional context and item factors, Adam for $2{,}200$ steps, learning rate $0.035$, and symmetric weight decay $2\times10^{-5}$.

The \emph{standard} model predicts with
\begin{equation}
  q_{\mathrm{std}}(i\mid c)
  =\softmax_i\left(H E^\top/\sqrt{d}\right).
\end{equation}
The \emph{prior-adjusted} model predicts with
\begin{equation}
  q_{\mathrm{PA}}(i\mid c)
  =\softmax_i\left(H E^\top/\sqrt{d}+\alpha\log\pi_i\right).
\end{equation}
Thus both models fit the same conditional distribution, but only the standard model must represent the prior inside $E$.

\paragraph{Metrics.}

We report the leading item-embedding energy $R_1(E)$, effective rank, and the absolute correlation between the first item principal-component projection and $\log\pi_i$. We also verify that both models fit the target distribution using $D_{\mathrm{KL}}(\pi_\alpha\|q)$.

To test the assumption behind \cref{prop:balanced-factorization}, let $\bar E=P_nE$ and let $Q$ contain an orthonormal basis for the active right-singular subspace of $\bar E$. We compare the gauge-correct factors $H_c=HQ$ and $E_c=\bar E Q$ using
\begin{equation}
  B_c=
  \frac{\|H_c^\top H_c-E_c^\top E_c\|_F}
  {\|H_c^\top H_c\|_F+\|E_c^\top E_c\|_F}.
  \label{eq:centered-balance}
\end{equation}
The projection excludes factor dimensions that cannot affect the centered score. Equal Frobenius norms alone would not establish balancedness. We also measure normalized spectrum-transfer error
\begin{equation}
  D_{\mathrm{spec}}
  =\frac{1}{2}\sum_k\left|
  \frac{s_k(E_c)^2}{\sum_j s_j(E_c)^2}
  -\frac{\sigma_k(ZP_n)}{\sum_j\sigma_j(ZP_n)}
  \right|,
  \label{eq:spectrum-transfer-error}
\end{equation}
which is zero for the balanced exact factorization. Finally, $\rho_{\mathrm{pop}}$ from \cref{eq:popularity-ratio}, with numerator scaled by $\alpha$, records when the sufficient spectral regime is reached.

\begin{figure}[t]
  \centering
  \includegraphics[width=\textwidth]{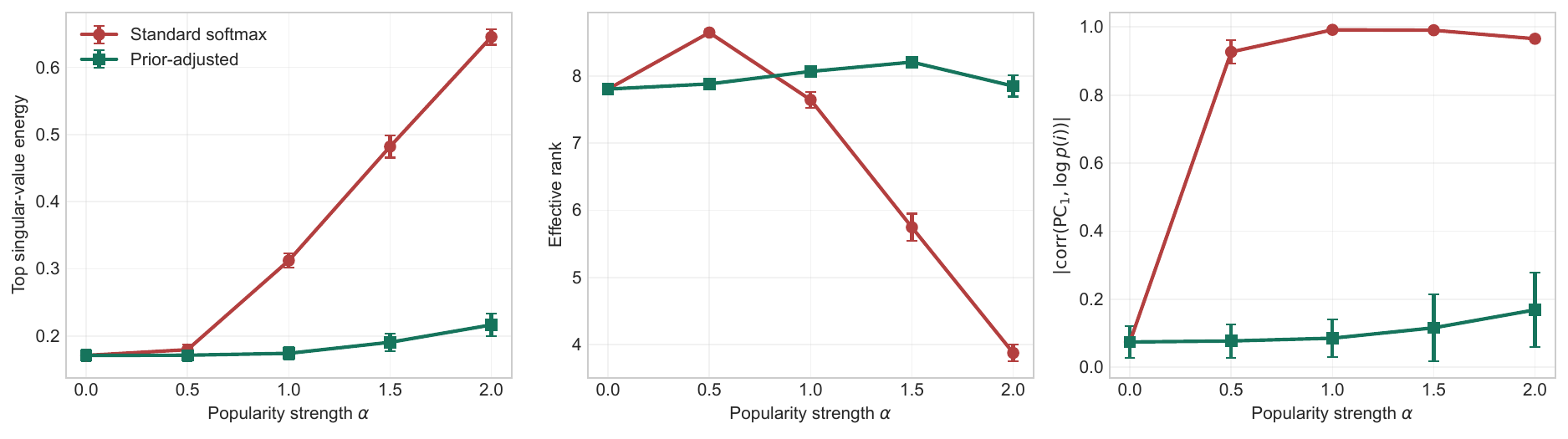}
  \caption{In the matched synthetic setting, popularity strength causes spectral concentration when the item prior must be represented inside the dot product. Error bars show one standard deviation over five seeds.}
  \label{fig:synthetic-results}
\end{figure}

\begin{table}[t]
  \centering
  \small
  \caption{Synthetic results, mean $\pm$ standard deviation over five seeds. PA denotes prior-adjusted softmax. KL values are means; their across-seed standard deviations are below $3\times10^{-4}$.}
  \label{tab:synthetic-results}
  \begin{tabular}{llrrrr}
    \toprule
    $\alpha$ & Method & $R_1(E)$ & $\erank(E)$ & $|\mathrm{corr}(\mathrm{PC}_1,\log\pi)|$ & KL \\
    \midrule
    0 & Standard & $0.171\pm0.008$ & $7.81\pm0.03$ & $0.074\pm0.047$ & $0.0004$ \\
    0 & PA       & $0.171\pm0.008$ & $7.81\pm0.03$ & $0.074\pm0.047$ & $0.0004$ \\
    1 & Standard & $0.312\pm0.010$ & $7.65\pm0.12$ & $0.992\pm0.002$ & $0.0018$ \\
    1 & PA       & $0.174\pm0.009$ & $8.07\pm0.03$ & $0.086\pm0.054$ & $0.0013$ \\
    2 & Standard & $0.645\pm0.012$ & $3.87\pm0.12$ & $0.965\pm0.002$ & $0.0037$ \\
    2 & PA       & $0.216\pm0.017$ & $7.85\pm0.16$ & $0.169\pm0.109$ & $0.0028$ \\
    \bottomrule
  \end{tabular}
\end{table}

\begin{table}[t]
  \centering
  \small
  \caption{Gauge-correct transfer diagnostics, mean $\pm$ standard deviation over five seeds. $\rho_{\mathrm{pop}}$ is averaged across generated problems and is shared by the two fitting methods. Lower $B_c$ and $D_{\mathrm{spec}}$ indicate closer agreement with balanced transfer.}
  \label{tab:factorization-diagnostics}
  \begin{tabular}{llrrr}
    \toprule
    $\alpha$ & Method & $\rho_{\mathrm{pop}}$ & $B_c$ & $D_{\mathrm{spec}}$ \\
    \midrule
    0 & Standard/PA & $0.00$ & $(2.20\pm0.59)\times10^{-5}$ & $(3.58\pm1.41)\times10^{-6}$ \\
    1 & Standard & $2.39$ & $0.00207\pm0.00051$ & $0.000231\pm0.000060$ \\
    1 & PA       & $2.39$ & $0.00017\pm0.00016$ & $0.000015\pm0.000010$ \\
    2 & Standard & $4.78$ & $0.0806\pm0.0045$ & $0.0136\pm0.0012$ \\
    2 & PA       & $4.78$ & $0.00031\pm0.00037$ & $0.000028\pm0.000021$ \\
    \bottomrule
  \end{tabular}
\end{table}

\paragraph{Results.}

At $\alpha=0$, the two objectives coincide and produce identical geometry. As popularity strengthens, standard softmax develops the predicted dominant direction. At $\alpha=1$, $\rho_{\mathrm{pop}}=2.39$ exceeds the sufficient threshold in \cref{prop:spectral-gap}, and the first principal-component projection correlates with $\log\pi_i$ at $0.992$. At $\alpha=2$, the leading direction contains $64.5\%$ of item-embedding energy and effective rank falls to $3.87$.

Prior adjustment largely removes this transition. Even at $\alpha=2$, leading energy is $21.6\%$ and effective rank is $7.85$. Both methods achieve small KL divergence, so the geometric difference is not caused by one model failing to fit the data. The prior-adjusted model fits popularity through the supplied offset and allocates substantially less leading-factor energy to it.

The new diagnostics support the conditional transfer argument in this setting. For standard softmax at $\alpha=2$, the centered Gram mismatch is $B_c=0.081$ and the normalized spectrum-transfer error is only $D_{\mathrm{spec}}=0.014$. The factors are therefore not exactly minimum-norm, but their spectra remain close to the balanced prediction. Prior adjustment has still smaller errors because its learned dot product need only factorize the semantic residual. These measurements validate the transfer condition for this experiment; they do not establish that arbitrary encoders and optimizers will be similarly balanced.

The non-monotonic standard effective rank between $\alpha=0$ and $0.5$ is expected: a weak additional direction can initially increase utilized rank before it becomes dominant. The sharp rise in leading energy and principal-component correlation beyond that point is the relevant signature. Because the supplied offset uses the exact prior form in \cref{eq:synthetic-data}, this experiment does not test finite-sample prior estimation, exposure confounding, or temporal misspecification.

\subsection{Alibaba/Tianchi public interaction-log study}
\label{sec:real-data}

\paragraph{Data and protocol.}
We use item-view events from the public Alibaba/Tianchi mobile recommendation competition \citep{tianchi2015mobile}, distributed with DPA \citep{lin2021dpa}. All item filtering statistics and priors are computed from events before the strict split at 2014-12-18 00:00. A deterministic medium-scale subset samples users with at least 50 training views and one held-out view, retains the most frequent eligible items, and iterates the filters to convergence. Post-split activity is used only to retain users with evaluable targets, not to estimate training frequencies; conditioning on evaluable users nevertheless limits population-level interpretation. Duplicate views are preserved, and previously viewed items are not masked at evaluation, so the task is repeat-view next-item retrieval rather than novel-item recommendation. The resulting data contain 459,105 training events over 30 days, 12,897 held-out events, 1,413 users, and 20,000 items. Item counts range from 11 to 290 (median 18); the most frequent 1\% of items account for 4.97\% of training events. This medium subset is less concentrated than the corresponding full strict data, whose top 1\% account for 7.39\%.

We compare two normalized 128-dimensional user--item factor models under identical initialization families, Adam settings, batch size 512, score scale $s=20$, five epochs, and seeds 2021--2023. Standard full softmax uses logits $s h_u^\top e_i$. Static prior-adjusted softmax uses $s h_u^\top e_i+\log\widehat\pi_i$, where $\widehat\pi$ is estimated only from training events with additive smoothing. Exact normalization over all 20,000 items avoids a sampled-negative confound. Every run processes all $2{,}295{,}525$ epoch-events and drops none. The user-ID encoder is deliberately simple: the comparison isolates the decoder objective rather than claiming a cross-architecture prevalence result.

\paragraph{Score-channel diagnostics.}
Let $S_c=sHE^\top P_n$ be the item-centered learned score matrix and
\begin{equation}
  \ell=\frac{P_n\log\widehat\pi}{\|P_n\log\widehat\pi\|_2}
\end{equation}
the normalized popularity direction. Write $a=S_c\ell$ and $\bar a=m^{-1}\one_m^\top a$. We distinguish the complete projection onto the popularity direction,
\begin{equation}
  S_\ell=a\ell^\top,
  \qquad
  \rho_\ell=\frac{\|S_\ell\|_F^2}{\|S_c\|_F^2},
  \label{eq:empirical-pop-direction}
\end{equation}
from the stricter context-shared component predicted by the static theory,
\begin{equation}
  \widehat S_{\mathrm{pop}}=\one_m\bar a\ell^\top,
  \qquad
  \widehat\rho_{\mathrm{pop}}
  =\frac{\|\widehat S_{\mathrm{pop}}\|_F^2}{\|S_c\|_F^2}.
  \label{eq:empirical-rank-one}
\end{equation}
The ratio $\widehat\rho_{\mathrm{pop}}/\rho_\ell$ measures how much popularity-direction energy has a coefficient shared across contexts. This least-squares projection is an operational context-shared popularity-aligned component, not an identified recovery of the theoretical prior matrix: context-mean PMI can also project onto $\ell$. Its response to the matched intervention is therefore evidence about learned-score directionality, while causal attribution requires additional assumptions. We additionally report $|v_1^\top\ell|$ for the leading right singular vector of $S_c$, the leading score-energy ratio $R_1(S_c)$, and the reduction in the leading singular value after subtracting $\widehat S_{\mathrm{pop}}$. These quantities are computed through the low-dimensional factors without materializing the full score matrix.

\begin{table}[t]
  \centering
  \small
  \caption{Score-channel diagnostics on Alibaba/Tianchi logs, mean $\pm$ standard deviation over three paired seeds. PA supplies the training marginal explicitly. The final column is the paired relative change from standard to PA; every listed paired delta has the same sign in all three seeds.}
  \label{tab:real-score-channel}
  \begin{tabular}{lrrr}
    \toprule
    Diagnostic & Standard & PA & Relative change \\
    \midrule
    $R_1(S_c)$ & $0.08820\pm0.00220$ & $0.04806\pm0.00171$ & $-45.5\%$ \\
    $|v_1^\top\ell|$ & $0.52721\pm0.00869$ & $0.08735\pm0.02173$ & $-83.4\%$ \\
    $\rho_\ell$ & $0.02604\pm0.00032$ & $0.001041\pm0.000265$ & $-96.0\%$ \\
    $\widehat\rho_{\mathrm{pop}}$ & $0.01905\pm0.00025$ & $0.000265\pm0.000105$ & $-98.6\%$ \\
    $\widehat\rho_{\mathrm{pop}}/\rho_\ell$ & $0.73150\pm0.01335$ & $0.24715\pm0.04343$ & $-66.2\%$ \\
    Leading-$\sigma$ removal effect & $0.13514\pm0.00376$ & $0.00272\pm0.00107$ & $-98.0\%$ \\
    \bottomrule
  \end{tabular}
\end{table}

\paragraph{Predictive fit and direction specificity.}
We evaluate held-out negative log-likelihood using each model's complete training-objective logits: $s h_u^\top e_i$ for standard softmax and $s h_u^\top e_i+\log\widehat\pi_i$ for static prior adjustment. Standard softmax obtains $8.85524\pm0.01941$ nats per event and prior adjustment obtains $8.90428\pm0.01307$. The paired difference is $0.04903\pm0.00646$ nats, or $0.55\%$ of the standard mean. A paired user-cluster bootstrap, resampling held-out users within each of the three fixed seeds, gives $[0.04313,0.05472]$. No equivalence margin was pre-specified; this is a small detected likelihood degradation, not an equivalence result.

To test whether the score-energy change is specific to the empirical popularity direction, we uniformly permute item labels of the exact centered log-frequency vector 2,000 times per seed. Each permutation preserves its values and norm while breaking correspondence with item identity. The observed standard-to-PA drop in strict shared energy is $0.01878$ on average, compared with a permutation-null mean of $2.12\times10^{-6}$. The true direction exceeds all 2,000 permutations in every seed, giving a one-sided empirical $p=1/2001$ per seed. This establishes item-direction specificity against the stated null; it does not distinguish the additive prior from PMI structure aligned with frequency.

\paragraph{Geometry and retrieval.}
The score intervention is accompanied by a change in item geometry. Standard softmax has $R_1(E)=0.02187\pm0.00077$ and mean pairwise cosine $0.02698\pm0.00395$. The absolute correlation between item PC1 and log popularity is $0.56802\pm0.00854$. Prior adjustment reduces these values to $0.01571\pm0.00019$, $0.00648\pm0.00161$, and $0.07806\pm0.02501$, respectively, while effective rank increases from $125.38\pm0.13$ to $126.12\pm0.06$. Neither model exhibits severe low-rank collapse: the effective ranks are close to the 128-dimensional maximum. Moreover, centered balance error is $0.858\pm0.0003$ for standard softmax and $0.864\pm0.0003$ for prior adjustment. The real-data factors are therefore far from the balanced regime of \cref{prop:balanced-factorization}; their embedding changes are empirical outcomes of this intervention, not validation of the transfer theorem.

For retrieval, standard softmax ranks by its learned dot product. The prior-adjusted model restores a seven-day training-only log marginal with coefficient $\beta=1$. Table~\ref{tab:real-retrieval} reports the pre-specified comparison at $K=50$. Confidence intervals resample held-out users within each seed and average the three fixed-seed estimates; they quantify user-level uncertainty, not generalization over optimizer seeds. The Macro Recall@50 point difference is positive and its interval includes zero, which is not a formal non-inferiority result. NDCG@50 increases, but explicit prior restoration reduces coverage and novelty. The exposed coefficient $\beta$ therefore controls a visible serving trade-off that standard softmax partly hides in its factors.

\begin{table}[t]
  \centering
  \small
  \caption{Alibaba/Tianchi retrieval results at $K=50$, mean $\pm$ standard deviation over three seeds. Bracketed intervals are paired user-bootstrap 95\% intervals for PA minus standard. Bucketed event Recall, coverage, and novelty differences are descriptive.}
  \label{tab:real-retrieval}
  \begin{tabular}{lrrr}
    \toprule
    Metric & Standard & PA + restored prior & Paired difference \\
    \midrule
    Macro Recall@50 & $0.36211\pm0.00405$ & $0.36553\pm0.00119$ & $0.00342\;[-0.00100,0.00803]$ \\
    Macro NDCG@50 & $0.18662\pm0.00189$ & $0.21350\pm0.00047$ & $0.02688\;[0.02320,0.03050]$ \\
    Head event Recall@50 & $0.42968\pm0.00229$ & $0.44529\pm0.00196$ & $0.01561$ \\
    Middle event Recall@50 & $0.34519\pm0.00341$ & $0.34154\pm0.00363$ & $-0.00366$ \\
    Tail event Recall@50 & $0.27035\pm0.00307$ & $0.27568\pm0.00113$ & $0.00533$ \\
    Catalog coverage & $0.98482\pm0.00044$ & $0.72612\pm0.00111$ & $-0.25870$ \\
    Novelty (bits) & $14.1765\pm0.0004$ & $14.0389\pm0.0037$ & $-0.1376$ \\
    \bottomrule
  \end{tabular}
\end{table}

The matched intervention is the central empirical test. Under standard softmax, the operational context-shared component accounts for $1.905\%$ of total centered score energy and $73.2\%$ of the energy projected onto the popularity direction. Although the global fraction is not dominant, subtracting this single component reduces the leading singular value by $13.5\%$. Supplying the same statistical prior outside the dot product reduces $98.6\%$ of the shared component and $96.0\%$ of the broader popularity-direction energy. The permutation null shows that this change is specific to the item identity of the observed log-frequency direction, while the NLL comparison reveals a small predictive cost. Together these results support a controllable popularity-aligned rank-one component in learned scores; they neither identify a unique causal decomposition nor imply that the complete score matrix is rank one.

\section{Discussion and Practical Guidance}
\label{sec:discussion}

\paragraph{What is and is not Transformer-specific.}
The derivation depends on the conditional likelihood, item marginal, and dot-product decoder. It applies when $h_c$ comes from a Transformer, GRU, CNN, graph network, or nonparametric pooling rule. Architectures still matter through expressivity, normalization, residual connections, and optimization. They can amplify, suppress, or add to the observable symptom, but replacing a Transformer with another encoder does not by itself remove the rank-one target channel.

\paragraph{Score anisotropy versus embedding anisotropy.}
The row-centered optimal score matrix has a gauge-invariant popularity component. A particular embedding matrix is not uniquely determined because $HE^\top=(HA)(EA^{-\top})^\top$. Our transfer result therefore requires balanced minimum-norm factors. Likewise, a narrow cone measured by positive mean cosine is not invariant to translation or row-logit gauge. Empirical studies should report predictive score structure, centered-factor balancedness, and embedding geometry rather than treating one visualization as a theorem.

\paragraph{Popularity is not always noise.}
The observed marginal $\pi_i$ contains exposure, availability, repeated behavior, and genuine aggregate preference. It is also partly generated by earlier serving policies. Removing it entirely may hurt calibrated next-item likelihood or business metrics. Prior separation is useful because it makes the trade-off explicit. At retrieval time one can rank with
\begin{equation}
  r_{ci}=h_c^\top e_i+\beta\log\pi_i,
\end{equation}
where $\beta$ is selected for the desired balance among accuracy, novelty, coverage, and tail exposure. In our static study, prior adjustment raises complete-logit held-out NLL by $0.0490$ nats per event ($0.55\%$), even though the primary Recall@50 difference includes zero and NDCG@50 increases under recent-prior restoration. This small detected likelihood cost may reflect finite optimization, model constraints, or temporal mismatch; the experiment does not identify which. The important change is not necessarily setting $\beta=0$, but making the observed-frequency contribution explicit and controllable rather than leaving it entangled in the dot product.

\paragraph{Static geometry and dynamic recommendation.}
DPA previously showed that time-local popularity can be used as an NCE noise distribution during training and restored from recent interactions during retrieval \citep{lin2021dpa}. \Cref{prop:dynamic-popularity} supplies the geometric bridge: each period contributes a rank-one centered direction, and the complete trajectory occupies the span of those directions. Our static intervention supports the corresponding one-marginal score mechanism, but the daily marginals in the real-data subset are not themselves close to rank one across time: the 30-day log-prior matrix has algebraic rank 30 and effective rank 27.43. Dynamic popularity can therefore occupy a broad temporal subspace even though each period contributes a rank-one channel. Comparing dynamic prior adjustment and DPA with score-channel diagnostics remains necessary.

\paragraph{Evaluation protocol.}
Overall Recall or NDCG can hide head--tail trade-offs. We report event Recall separately for head, middle, and tail items, together with catalog coverage and novelty. In the primary comparison, head and tail Recall@50 increase by $0.01561$ and $0.00533$, respectively, while middle Recall decreases by $0.00366$; these subgroup differences are descriptive rather than separately inferential. Geometry should be monitored with $R_1(E)$, effective rank, mean pairwise cosine, correlation between principal components and $\log\pi_i$, and the score-level quantities in \cref{eq:empirical-pop-direction,eq:empirical-rank-one}. Restoring the prior at $\beta=1$ reduces coverage and novelty, and the sweep in \cref{tab:beta-sensitivity} shows that the best coefficient depends on the serving objective. A mitigation is successful only relative to an explicitly chosen relevance--exposure trade-off.

\paragraph{Limitations.}
The synthetic generator contains the exact global $\alpha\log\pi_i$ term supplied to the prior-adjusted model, so its positive result does not establish robustness to prior misspecification. The Alibaba/Tianchi study estimates its priors from training events, but uses a deterministic medium subset rather than the complete log. Its top-1\% event share is 4.97\%, compared with 7.39\% in the full strict data, and three seeds establish repeatability rather than high-powered seed-level inference. Subset membership is conditioned on having a post-split target, and evaluation allows repeated previously viewed items; both choices restrict generalization beyond repeat-view retrieval for active users. The user-ID factor model isolates the dot-product objective; it does not establish how often the component dominates in Transformer, recurrent, graph, or production encoders. The high centered balance error also means that its embedding changes do not test the minimum-norm transfer result. The operational shared projection is not an identified decomposition because mean PMI can align with frequency. The permutation null establishes item-direction specificity, not causal separation of prior, exposure, and preference. View frequency itself mixes exposure, availability, repeated browsing, and preference. Real systems add approximate retrieval, tied embeddings, sampled negatives, nonlinear encoder constraints, and feedback loops. The sufficient spectral condition in \cref{prop:spectral-gap} can be conservative, and balanced factorization is a measured condition in the synthetic study rather than a universal property of optimization. Cross-encoder and full-log studies are needed before making prevalence claims.

\section{Conclusion}

Representation degeneration in sequential recommendation should not be attributed to Transformer encoders alone. Under a dot-product softmax objective, the population-optimal logit contains an item-marginal term $\log\pi_i$. Across contexts, a static marginal contributes a gauge-invariant rank-one component whose strength grows with long-tail severity; time-varying marginals span a rank-at-most-$T$ temporal subspace. When a low-rank model learns an approximately balanced factorization of such scores, dominant item-embedding directions are a natural outcome. Other architectural and optimization mechanisms may coexist or dominate in a particular system.

Non-target gradients explain how rare items can follow shared, noisy trajectories in finite data, but their expectation cancels at the calibrated population optimum. This connects the gradient and population views without dismissing finite-time effects. Our earlier DPA work used dynamic NCE target shaping and prediction recovery to separate popularity operationally; the present analysis explains the score geometry targeted by this separation and extends the principle to a continuous prior-adjusted full-softmax objective. In Alibaba/Tianchi logs, a matched intervention reduces the operational context-shared popularity-aligned energy by 98.6\% and the broader popularity-direction energy by 96.0\%. The reduction exceeds all 2,000 item-label permutations in every seed, but incurs a small detected held-out NLL cost of $0.0490$ nats per event; the Recall@50 interval includes zero. Corrected sampling, stronger tail supervision, and uniformity regularization remain complementary responses. The resulting claim is precise: the rank-one marginal component exists in population-optimal scores independently of the encoder, while its learned magnitude, causal interpretation, and transfer to embedding geometry depend on data, factorization, and optimization.

\section*{Research Declarations}

\paragraph{Data and code availability.}
The synthetic experiment, real-data preprocessing and evaluation pipeline, aggregate measurements, and commands needed to reproduce the reported tables are included in the source archive. The raw Alibaba/Tianchi interactions are available from the competition data page at \url{https://tianchi.aliyun.com/competition/entrance/231522/information} and are not redistributed here. The historical DPA implementation is available at \url{https://github.com/alibaba/Dynamic-popularity-aware-recommendation}.

\paragraph{Ethics declaration.}
This study performs secondary analysis of a released recommendation competition dataset and introduces no new human participants or interventions. The analysis uses only the opaque identifiers provided with the data and makes no attempt to identify users.

\paragraph{Author contributions.}
Yang Cheng: Conceptualization, Methodology, Software, Validation, Formal Analysis, Investigation, Data Curation, Writing--Original Draft, Writing--Review and Editing, and Visualization.

\paragraph{Conflict of interest.}
The author declares no conflict of interest.

\paragraph{Funding.}
This research received no external funding.

\paragraph{AI-tool disclosure.}
Generative AI tools, including Claude and OpenAI Codex, were used for language editing, code critique, and experiment orchestration. The author inspected the derivations, citations, implementation, and reported results and takes responsibility for the manuscript.

\bibliographystyle{abbrvnat}
\bibliography{references}

\appendix
\section{Proof of the Balanced Factorization Result}
\label{app:balanced-proof}

Let $\widetilde Z=U\Sigma V^\top$ be a compact singular-value decomposition. The factorization
\begin{equation}
  H=U\Sigma^{1/2},
  \qquad
  E=V\Sigma^{1/2}
\end{equation}
is feasible and has objective value
\begin{equation}
  \frac{1}{2}(\|H\|_F^2+\|E\|_F^2)
  =\operatorname{tr}(\Sigma)
  =\|\widetilde Z\|_*.
\end{equation}
For any feasible $H,E$, the variational characterization of the nuclear norm gives
\begin{equation}
  \|HE^\top\|_*
  \leq \|H\|_F\|E\|_F
  \leq \frac{1}{2}(\|H\|_F^2+\|E\|_F^2).
\end{equation}
Therefore no factorization can achieve a smaller value. Equality requires balanced factors in the singular subspaces, giving
\begin{equation}
  H=U\Sigma^{1/2}R,
  \qquad
  E=V\Sigma^{1/2}R
\end{equation}
for an orthogonal matrix $R$. Right multiplication by $R$ does not change singular values. Thus $s_k(E)=\sqrt{\sigma_k(\widetilde Z)}$, which yields \cref{eq:embedding-energy}.

\section{Reproducibility Details}
\label{app:reproducibility}

The synthetic study requires Python 3, NumPy, and Matplotlib. From the source directory, run
\begin{verbatim}
python3 experiment.py --output-dir figures --steps 2200 --seeds 5
\end{verbatim}
The script writes per-run measurements to \texttt{figures/synthetic\_runs.csv}, aggregate statistics to \texttt{figures/synthetic\_summary.json}, and the paper figure in PDF and PNG formats. All random generators are explicitly seeded. Exact expected cross-entropy is evaluated over all $64\times256$ context--item pairs at every optimization step.

\section{Real-Data Reproducibility Details}
\label{app:real-reproducibility}

The source archive contains the PyTorch pipeline in \texttt{real\_data\_code/} and the machine-readable aggregate measurements in \texttt{artifacts/real\_data/}. The archived raw CSV used for the reported run has SHA-256
\begin{verbatim}
52b9fdfcf3793dbf827a2ca752f9c6cc04fe9407392e3829b9a175d86bd31c18
\end{verbatim}
After placing the Tianchi CSV at \texttt{data/fresh\_comp\_offline/}, the processed datasets are generated by
\begin{verbatim}
python real_data_code/preprocess.py \
  --raw-csv data/fresh_comp_offline/tianchi_fresh_comp_train_user.csv \
  --output-dir data/processed/strict --protocol strict

python real_data_code/subsample.py \
  --source-dir data/processed/strict \
  --output-dir data/processed/medium \
  --num-users 2000 --max-items 20000 \
  --min-user-train 50 --min-user-test 1 --min-item-train 5 \
  --min-retained-user-fraction 0.7 --seed 2021
\end{verbatim}
The matched matrix and aggregation commands are
\begin{verbatim}
python real_data_code/run_matrix.py \
  --data-dir data/processed/medium --runs-dir real_data_runs \
  --methods full_softmax static_pa --seeds 2021 2022 2023 \
  --epochs 5 --device auto --skip-completed

python real_data_code/aggregate.py \
  --runs-dir real_data_runs --output-dir real_data_results \
  --primary-k 50 --bootstrap-samples 10000

python real_data_code/diagnose_fit_specificity.py \
  --data-dir data/processed/medium --runs-dir real_data_runs \
  --output-dir fit_specificity_results \
  --permutations 2000 --bootstrap-samples 10000
\end{verbatim}
The six unit tests cover temporal splitting, probability normalization, batching, fit-specificity helpers, the uniformity gradient, and the synthetic rank-one channel. Every reported training run records its processed and dropped example counts in \texttt{config.json}; the six runs used here report zero dropped examples.

\clearpage
\section{Serving-Prior Sensitivity}

\begin{table}[ht]
  \centering
  \small
  \caption{Static prior-adjusted retrieval at $K=50$ as the restored recent-prior coefficient $\beta$ varies. Values are means over three seeds. The sweep is descriptive and exposes the serving trade-off rather than selecting a universally optimal coefficient.}
  \label{tab:beta-sensitivity}
  \begin{tabular}{rrrrr}
    \toprule
    $\beta$ & Macro Recall & Macro NDCG & Coverage & Novelty (bits) \\
    \midrule
    0.00 & 0.35215 & 0.16986 & 0.99462 & 14.4459 \\
    0.25 & 0.36809 & 0.19622 & 0.92832 & 14.3290 \\
    0.50 & 0.36840 & 0.20539 & 0.80362 & 14.2219 \\
    0.75 & 0.36535 & 0.21003 & 0.73942 & 14.1279 \\
    1.00 & 0.36553 & 0.21350 & 0.72612 & 14.0389 \\
    \bottomrule
  \end{tabular}
\end{table}

\section{Remaining Real-Data Ablations}

The present experiment identifies the static mechanism but leaves four important extensions. First, dynamic prior-adjusted softmax should be compared with DPA-style dynamic NCE under identical time windows. Second, the complete 5.19-million-event strict dataset should replace the medium subset. Third, Transformer, recurrent, graph, and pooling encoders should be compared while holding the decoder and supervision fixed. Fourth, temporal-window, smoothing, and exposure-propensity ablations should test prior misspecification. These studies should retain overall and head/middle/tail Recall@K and NDCG@K, catalog coverage, novelty, the diagnostics in \cref{eq:empirical-pop-direction,eq:empirical-rank-one}, factor balancedness, and the singular spectrum of the daily log-marginal matrix.

\end{document}